\documentclass[sn-mathphys-num]{sn-jnl}

\usepackage{graphicx}%
\usepackage{multirow}%
\usepackage{amsmath,amssymb,amsfonts}%
\usepackage{amsthm}%
\usepackage{mathrsfs}%
\usepackage[title]{appendix}%
\usepackage{xcolor}%
\usepackage{textcomp}%
\usepackage{manyfoot}%
\usepackage{booktabs}%
\usepackage{algorithm}%
\usepackage{algorithmicx}%
\usepackage{algpseudocode}%
\usepackage{listings}%
\usepackage{lmodern}
\usepackage{verbatim}
\usepackage{chemarr}
\usepackage{color}
\usepackage{textgreek}

\def\bc{{\boldsymbol c}}

\def\bk{{\boldsymbol k}}

\def\bx{{\boldsymbol x}}

\def\bS{{\boldsymbol S}}

\def\bv{{\boldsymbol v}}

\newcommand{\R}{\mathbb{R}}

\usepackage{tikz}
\usetikzlibrary{decorations.markings,intersections,positioning,calc}

  \tikzset{mylabel/.style  args={at #1 #2  with #3}{
    postaction={decorate,
    decoration={
      markings,
      mark= at position #1
      with  \node [#2] {#3};
 } } } }

\theoremstyle{thmstyleone}%
\newtheorem{thm}{Theorem}
\newtheorem{conj}{Conjecture}%

\theoremstyle{thmstyletwo}%
\newtheorem{exa}{Example}%
\newtheorem{rem}{Remark}%

\theoremstyle{thmstylethree}%

\begin{document}



\title{When algebra twinks system biology: a conjecture on the structure of Gr\"obner bases in complex chemical reaction networks}

\author*[1,2]{\fnm{Paola} \sur{Ferrari}}\email{paola.ferrari@unige.it}

\author[1]{\fnm{Sara} \sur{Sommariva}}\email{sara.sommariva@unige.it}

\author[1,3]{\fnm{Michele} \sur{Piana}}\email{michele.piana@unige.it}

\author[1]{\fnm{Federico} \sur{Benvenuto}}\email{federico.benvenuto@unige.it}

\author[1]{\fnm{Matteo} \sur{Varbaro}}\email{matteo.varbaro@unige.it}

\affil[1]{\orgdiv{Department of Mathematics}, \orgname{University of Genoa}, \orgaddress{\street{Via Dodecaneso 35}, \city{Genoa}, \postcode{16145}, \country{Italy}}}

\affil[2]{\orgdiv{School of Mathematics and Natural Sciences}, \orgname{University of Wuppertal}, \orgaddress{\street{Gaussstrasse 20}, \city{Wuppertal}, \postcode{42119}, \country{Germany}}}

\affil[3]{\orgdiv{Life Science Computational Laboratory}, \orgname{IRCCS Ospedale Policlinico San Martino}, \orgaddress{\street{Largo Rosanna Benzi 10}, \city{Genoa}, \postcode{16132}, \country{Italy}}}


\abstract{
We address the challenge of identifying all real positive steady states in chemical reaction networks (CRNs) governed by mass-action kinetics. Gröbner bases offer an algebraic framework that systematically transforms polynomial equations into simpler forms, facilitating comprehensive enumeration of solutions.
In this work, we propose a conjecture that CRNs with at most pairwise interactions yield Gröbner bases possessing a near-``triangular'' structure under appropriate assumptions, and we establish this triangularity rigorously under a strengthened set of hypotheses. We illustrate the phenomenon using examples from a gene regulatory network and the Wnt signaling pathway, where the Gröbner basis approach reliably captures all real positive solutions. Our computational experiments reveal the potential of Gröbner bases to overcome limitations of local numerical methods for finding the steady states of complex biological systems, making them a powerful tool for understanding dynamical processes across diverse biochemical models.
}

\maketitle

\section*{Introduction}

The determination of nonnegative steady states in large systems of quadratic equations is a fundamental problem in various fields, including systems biology, chemical engineering, and applied mathematics. These systems often arise from the modeling of complex chemical reaction networks (CRNs), where the concentrations of species evolve according to quadratic ordinary differential equations (ODEs) derived from the law of mass action \cite{Feinberg1987,Voit2013}.

Traditional numerical methods for solving these nonlinear systems include iterative approaches like Newton's method and gradient descent \cite{Berra2024404}. While these methods can be effective when a good initial approximation is available, they are inherently local and may fail to converge to a solution if the starting point is not sufficiently close to it. Moreover, these methods typically find a single solution, potentially missing other relevant solutions, especially in the context of multiple steady states that can exist in biological systems due to bistability or multistability phenomena \cite{Angeli2007}.

In the context of CRNs, particularly those modeling cell signaling pathways involved in diseases like cancer, finding all possible steady states is essential. For example, cells may become cancerous depending on specific genetic mutations, involving hundreds of proteins and reactions \cite{Sommariva2023}. 
Accurately determining all steady states allows researchers to understand the different potential behaviors of the system under various conditions, which is critical for developing targeted therapies.

To overcome the limitations of iterative methods, we adopt an algebraic approach based on Gr\"obner bases \cite{Cox2015,Buchberger2006,Bardet201549,KR2000}. Gr\"obner bases provide a powerful tool for solving systems of polynomial equations by transforming them into a simpler equivalent system, from which all solutions can be systematically derived. 
This method does not require initial guesses and can determine the exact number of solutions, including all real positive solutions relevant for concentrations in CRNs.

{Typically, applying this method to systems of equations generated by reaction networks involving at most pairwise interactions - which lead to quadratic algebraic systems - yields Gr\"obner bases with a distinctive structure: the associated ideal contains a univariate polynomial whose degree does not exceed the system size, and the remaining generators form an almost triangular basis.
In this paper, we hypothesize that this phenomenon is not accidental and we propose a conjecture to explain this observation, which we prove under a strengthened set of hypotheses. In particular, we require a condition on the Jacobian of the vector field defined by the independent mass–action balances and the conservation laws, a test that can be computationally onerous for large networks. Lighter and readily applicable conditions under which the triangular property holds true have to be discussed and further analyzed.

The first application of Gr\"obner bases for computing steady-state concentrations in a simple enzyme channeling model dates back to the 1990s \cite{bayram1997novel}. 
This study shows the conjectured phenomenon, yet the work was not further developed despite significant advancements in computational capabilities over the years.
Nevertheless, the interest in using Gr\"obner bases in systems biology has persisted in recent literature, and some researchers have revisited their use to compute steady-state ideals and conducted intriguing studies to drastically improve computation speed by assigning a specific variable order
\cite{sadeghimanesh2019groebner,feliu2024toricity}.
}

All Macaulay2 scripts that compute the Gröbner bases used to test our conjecture are available in a companion GitHub repository \cite{M2Repository}.  
We checked the conjecture on many reaction-network examples: an ERK pathway \cite{Loman2023}, a minimal oscillating subnetwork model of the MAPK cascade \cite{Hadac2017}, TGF-$\beta$ receptors \cite{SAMAL20163}, and toy models \cite{Berra2024404}, which are all stored in the repository. The two models analysed in detail here, namely a seven-species gene-regulatory network \cite{Conradi2017} and a nineteen-species Wnt signalling pathway \cite{Gross201621}, are representative excerpts from this larger test bed.

The paper is organized as follows. In Section \ref{sec:algorithm}, we present our conjecture on the Gr\"obner basis structure for polynomial systems arising from CRNs, outlining the necessary preliminaries in Subsection \ref{ssec:preliminaries}, introducing the problem, our conjecture and a proof under stronger assumptions in Subsection \ref{ssec:conjecture}, and describing a possible step-by-step procedure for determining all real positive system solutions in Subsection \ref{ssec:algorithm}. In Section \ref{sec:gene_regulatory}, we test this method using a gene regulatory network, for which we derive the steady-state equations, compute a Gr\"obner basis, and analyze the number of positive solutions as the parameters change. Section \ref{sec:wnt} applies the same approach to the Wnt signaling pathway, illustrating how to handle larger, more complex reaction networks. Finally, Section \ref{sec:conclusions} offers our conclusions and highlights potential directions for future investigation.

\section{Algorithm for Finding All Real Positive Solutions}\label{sec:algorithm}

In this section, we present an algorithmic approach to determine all real positive solutions of polynomial systems arising from CRNs governed by the law of mass action. The algorithm leverages Gr\"obner bases to transform the system into a form that is more suitable for systematic solution. We begin by outlining the assumptions and presenting a conjecture that underpins the algorithm. We then detail the steps of the algorithm, providing insights into its implementation.

\subsection{Preliminaries}\label{ssec:preliminaries}

Consider a CRN with $n$ chemical species whose concentrations $x_1,\dots,x_n$ follow mass-action kinetics.  A steady state is normally described by $n$ polynomial balance equations, but exactly $p$ of those equations are linearly dependent because of $p$ independent conservation laws.  Replacing the $p$ redundant balances by these $p$ linear conservation relations gives a square system of $n$ algebraically independent polynomials
\begin{equation}\label{eq:steady_state_eq}
  f_i(x_1,\dots,x_n)=0,\qquad i=1,\dots,n,
\end{equation}
where $f_1,\dots,f_{\,n-p}$ are net production rates for $n-p$ species and $f_{\,n-p+1},\dots,f_n$ enforce the conservation laws.  The remainder of this section builds the algebraic framework in which system \eqref{eq:steady_state_eq} will be analysed.

Let $R=\mathbb{R}[x_1,\ldots ,x_n]$ and $S=\mathbb{C}[x_1,\ldots ,x_n]$ be polynomial rings in $n$ variables, respectively, over the real and complex field. For $f_1,\ldots ,f_r$ polynomials of $R$, we denote $I\subseteq R$ the ideal they generate in $R$ and $IS$ the ideal they generate in $S$. The zero-loci of the set of polynomials $f_1,\ldots ,f_r$ in $\mathbb{R}^n$ and $\mathbb{C}^n$, since they only depend by the ideal they generate, will be denoted by $\mathcal{Z}(I)\subseteq \mathbb{R}^n$ and $\mathcal{Z}(IS)\subseteq \mathbb{C}^n$. 

The ideal $I$ is called {\bf $0$-dimensional} if the $\mathbb{R}$-vector space $S/I$ is non-zero and has finite dimension. This is equivalent to ask for the zero-locus $\mathcal{Z}(IS)$ being finite and non-empty.
It turns out that a zero-dimensional ideal is {\bf radical} if and only if $|\mathcal{Z}(IS)|=\dim_{\mathbb{R}}R/I$.
As a last piece of notation, we say that a zero-dimensional ideal $I$ is {\bf in a normal form} if, for all $p=(p_1,\ldots ,p_n), q=(q_1,\ldots ,q_n)\in\mathcal{Z}(IS)$ such that $p\neq q$, we have $p_1\neq q_1$. We have (e.g. see Theorem 3.7.25 in \cite{KR2000}):

\begin{thm}\label{t:shape}
With the above notation, the following are equivalent:
\begin{enumerate}
    \item $I$ is a $0$-dimensional radical ideal in normal form.
    \item The reduced Gr\"obner basis of $I$ w.r.t. the lexicographical monomial order extending the linear order of the variables $x_n>x_{n-1}>\ldots >x_1$ is of the form 
    \[x_n-g_n, \ \ldots , \ x_2-g_2, \ g_1\] 
    where $g_i\in\mathbb{R}[x_1]$ for all $i=1,\ldots ,n$ and $g_1$ has distinct roots in $\mathbb{C}$.
\end{enumerate}
\end{thm}

\begin{rem}
With the notation of the theorem, the fact that $g_1$ has distinct roots in $\mathbb{C}$ is equivalent to say that the degree og $g_1$ is equal to the number of solutions of the polynomial system of equations $f_1=\ldots =f_r=0$ over $\mathbb{C}$, that is $\deg(f_1)=|\mathcal{Z}(IS)|$. This is also equivalent to say that $g_1$ and its derivative have no common factor.
\end{rem}

\smallskip

If the ideal $I=(f_1,\ldots ,f_r)\subseteq R$ is $0$-dimensional we must have $r\geq n$. By Theorem \ref{t:shape}, a zero-dimensional {\it radical} ideal can indeed be generated by $n$ elements. This drastically fails dropping the ``radical assumption'', as shown by the ideal $I=(x_1^{r-1},x_1^{r-2}x_2,\ldots ,x_2^{r-1})\subseteq \mathbb{R}[x_1,x_2]$, that is a $0$-dimensional ideal in 2 variables but, as it is easy to check, cannot be generated by less than $r$ elements. 

In our situation, the system of polynomial equations already consists of $n$ elements, namely we start from $f_1,\ldots ,f_n$ polynomials of $R=\mathbb{R}[x_1,\ldots ,x_n]$. Unfortunately, this is not enough to assure we are in the nice situation of Theorem \ref{t:shape}:

\begin{exa}
If $n=3$ consider the following systems of 3 polynomials:
\begin{align*}
    A=\begin{cases}
    x_3+x_2x_1+x_1^2=0 \\
    -x_3+x_2^2+x_2x_1=0 \\
    x_2+x_1+1=0
    \end{cases}, \ \ \ B=\begin{cases}
    x_3x_2+x_3+x_1^2=0 \\
    x_3x_1+x_2^2+x_2=0 \\
    x_2+x_1+1=0
    \end{cases} \\
    C=\begin{cases}
    x_3^2+2x_3x_1-2x_3+2x_2^2=0 \\
    x_3x_1-x_3+x_2^2=0 \\
    x_1-1=0
    \end{cases}, \ \ \ D=\begin{cases}
    x_3-x_1-2=0 \\
    x_2^2+x_2x_1-x_2=0 \\
    x_1^2+x_1=0
    \end{cases}
\end{align*}
One can check that:
\begin{enumerate}
    \item the system $A$ does not admit solutions, i.e., if $I\subseteq \mathbb{R}[x_1,x_2,x_3]$ is the ideal generated by the three polynomials in the system, $R/I=0$ (equivalently $I=R$);
    \item the system $B$ admit infinite solutions, so the ideal $I\subseteq \mathbb{R}[x_1,x_2,x_3]$ generated by the 3 polynomials in the system, even if proper, is not zero-dimensional;
    \item the ideal $I\subseteq \mathbb{R}[x_1,x_2,x_3]$ generated by the 3 polynomials in the system $C$ is 0-dimensional, but not radical;
    \item the ideal $I\subseteq \mathbb{R}[x_1,x_2,x_3]$ generated by the 3 polynomials in the system $D$ is 0-dimensional and radical, but not in a normal form.
\end{enumerate}
Therefore, all the systems of polynomial equations above do not fall, for different reasons, in the framework of Theorem \ref{t:shape}.
\end{exa}

\subsection{The problem and the conjecture}\label{ssec:conjecture}

Throughout this section we focus on the polynomial systems defining the steady-state ideal of chemical reaction networks that satisfy the assumptions listed below.
\begin{enumerate}
    \item Mass-action kinetics: the CRN follows the law of mass action, so reaction rates are proportional to the product of reactant concentrations raised to their stoichiometric coefficients (see \cite[Subsection 2.1.2]{MR3890056}).
    \item Inclusion of Conservation Laws: any conservation relations (e.g., total mass or charge conservation) are included in the system as additional equations or incorporated into the existing equations.
    \item Absence of independent subnetworks: the reaction set cannot be partitioned into two non-trivial subnetworks whose stoichiometric subspaces form a direct sum; equivalently, the global stoichiometric subspace is indecomposable (see \cite[Appendix~6.A]{MR3890056}).
\end{enumerate}

We would like to achieve an algebraic interpretation of the above assumptions that allows us to characterize the CRNs satisfying Hypothesis 1 of Theorem \ref{t:shape}. Using the equivalence in Theorem \ref{t:shape}, we would have a neat algorithm to finding the real positive solutions of the given CRN, as described in \ref{ssec:algorithm}. According to the numerical experiments we performed, it is not straightforward to find a CRN satisfying the three assumptions above that does dot fulfil hypothesis 1. Extensive numerical searches (see \cite{M2Repository}) have yet to reveal a CRN that meets Assumptions 1–3 but violates Hypothesis 1. Motivated by this empirical evidence, we state the following conjecture.

\begin{conj}\label{conj}
Let the steady–state ideal of a CRN satisfy Assumptions 1–3 listed above.  
Then there exists a permutation $\sigma\in S_n$ such that, after re-labelling the variables 
\[
(x_{\sigma(1)},x_{\sigma(2)},\dots,x_{\sigma(n)}) ,
\]
the reduced Gr\"obner basis of the ideal, computed with respect to the lexicographic order
\[
x_{\sigma(n)} > x_{\sigma(n-1)} > \dots > x_{\sigma(2)} > x_{\sigma(1)},
\]
has the triangular “shape”
\[
x_{\sigma(n)} - g_n(x_{\sigma(1)}),\;
x_{\sigma(n-1)} - g_{n-1}(x_{\sigma(1)}),\;
\dots,\;
x_{\sigma(2)} - g_2(x_{\sigma(1)}),\;
g_1(x_{\sigma(1)}),
\]
where $g_1\in\mathbb{K}[x_{\sigma(1)}]$ and, for every $j\ge 2$, 
$g_j\in\mathbb{K}[x_{\sigma(1)}]$ as well (here $\mathbb{K}$ is the extension field of $\mathbb{Q}$ with the coefficients of the equations in the system).  In particular, the first polynomial is
univariate in $x_{\sigma(1)}$, and each remaining polynomial is linear in the
corresponding variable with coefficients that are polynomials in $x_{\sigma(1)}$.
\end{conj}

This conjecture implies that the Gr\"obner basis transforms the system into a triangular form, facilitating a sequential solution process starting from $x_{\sigma(1)}$.

Having stated the conjecture, we now present a sufficient algebraic framework for a proof.  We work with a concentration vector $\bx=(x_1,\dots,x_n)^{\!\top}\in\R_{>0}^n$ and a positive rate-constant vector $\bk=(k_1,\dots,k_m)^{\!\top}$, write the mass-action rate vector as $\bv(\bx,\bk)=(v_1,\dots,v_m)^{\!\top}$ whose entries are monomials in the $x_i$, and let $\bS$ be a row basis of the stoichiometric subspace. We have
\[\bS\,\bv(\bx,\bk)=\mathbf 0.\]

Conservation relations correspond to vectors $\boldsymbol{\gamma}\in\ker(\mathbf S^{\mathsf T})$.  Following \cite{sommariva2021}, choose linearly independent, component-wise non-negative generators $\{\boldsymbol{\gamma}_1,\dots,\boldsymbol{\gamma}_p\}$ of the convex cone
$\{\boldsymbol{\gamma}\in\ker(\mathbf S^{\mathsf T})\mid \boldsymbol{\gamma}\ge0\}$.
Collecting these rows defines
$$
\mathbf N\;=\;
\begin{bmatrix}
\boldsymbol{\gamma}_1^{\!\mathsf T}\\
\vdots\\
\boldsymbol{\gamma}_p^{\!\mathsf T}
\end{bmatrix}
\;\in\;\mathbb R^{p\times n}.
$$

%

Together with the $(n-p)$ independent mass–action balances we obtain the $n$–component vector field
$$
\mathbf f(\bx)\;=\;
\begin{bmatrix}
\mathbf S\,\bv(\bx,\bk)\\[2pt]
\mathbf N\bx-\bc
\end{bmatrix},
$$
whose first $n-p$ components are the independent reaction–rate equations and whose last $p$ components are the conservation laws. A steady state is therefore a common zero of $\mathbf f$, i.e. a solution of \eqref{eq:steady_state_eq}.

\begin{thm}\label{prop:hyp1}
Let $I\;=\;\bigl\langle\mathbf f(\bx)\bigr\rangle\subset\mathbb{K}[\bx]$ be the steady–state ideal of a CRN. 
Assume that
\[
\det\!(J)\;\neq\;0, \qquad J = \nabla_{\!\bx}\mathbf f(\bx^*),
\]
for every zero $\bx^*$. Then $I$
is zero–dimensional, radical, and, for almost all changes of variables, in {normal form}.  Consequently $I$ (up to a possibly changing of variables) satisfies Hypothesis~1 of Theorem~\textup{\ref{t:shape}}.
\end{thm}

\begin{proof}

By the Jacobian Criterion \cite[Theorem 16.19]{MR1322960}, for any $\mathbf{x}^* \in \mathcal{V}(I)$:
$$\dim_{\mathbf{x}^*} \mathcal{V}(I) \leq n - \text{rank}(J(\mathbf{x}^*)).$$
Since $\det J(\mathbf{x}^*) \neq 0$, we have $\text{rank}(J(\mathbf{x}^*)) = n$, so $\dim_{\mathbf{x}^*} \mathcal{V}(I) = 0$ at every point.
Therefore $\dim \mathcal{V}(I) = 0$, making $I$ zero-dimensional.

Concerning radicality, since $\det J(\mathbf{x}^*) \neq 0$ at every point and $I$ is zero-dimensional, $\mathbb{K}[\bx]/I$ is a 0-dimensional regular ring by \cite[Theorem 16.19 (b)]{MR1322960}, hence $I$ is radical. 

At this point we can use Proposition 4.2.2 di \cite{MR2363237} to infer that there exists a nonempty Zariski open subset $U\subset \mathbb{K}^{n-1}$ such that for all $\mathbf a =(a_1,\ldots ,a_{n-1})\in U$ a change of variables the of the form:
$$
\varphi_{\mathbf a} :\;x_{\sigma(i)}\mapsto x_{\sigma(i)}\;(i>n),\qquad 
x_{\sigma(1)}\mapsto x_{\sigma(1)}+\sum_{i=2}^{n}a_{i}x_{\sigma(i)},
$$
will put $I$ in normal form.

\end{proof}

For large networks, verifying the condition
$$
\det\!\bigl[\nabla_{\!\bx}\mathbf f(\bx^*)\bigr]\neq0
$$
symbolically may be computationally expensive: the Jacobian is an $n\times n$ polynomial matrix and one must check that the resulting determinant does not vanish on the whole steady-state variety.  Nevertheless the hypothesis is pragmatically plausible for many biochemical models.  First, non-degeneracy is generic in parameter space: the set of rate constants for which the Jacobian drops rank is contained in a proper algebraic subset, so any small perturbation of a ``bad’' parameter point typically restores full rank.  Second, recent works, such as \cite{Biddau2024}, adopt exactly the same requirement, providing further evidence that the condition aligns with the behaviour of many real-world CRNs.

\subsection{Algorithm Description}\label{ssec:algorithm}

The algorithm proceeds as follows:

\begin{enumerate}
    \item \textbf{Compute the Gr\"obner Basis:}
    \begin{itemize}
        \item Compute the reduced Gr\"obner basis $\{g_1, g_2, \dots, g_n\}$ of the original system $\{f_1, f_2, \dots, f_n\}$ using lexicographic ordering with $x_n > x_{n-1} > \dots > x_2 > x_1$.
        \item The computation can be performed using algorithms like Buchberger's algorithm or more efficient variants (e.g., F4 or F5 algorithms).
    \end{itemize}
    
    \item \textbf{Solve the Univariate Polynomial Equation:}
    \begin{itemize}
        \item Solve the univariate polynomial equation $g_1(x_1) = 0$ to find all real positive roots of $x_1$.
        \item Utilize robust numerical methods suitable for univariate polynomials, such as the multiprecision algorithm by Bini and Robol \cite{Bini2014}.
        \item Denote the set of real positive solutions as $\{x_1^{(k)}\}$, where $k$ indexes the solutions.
    \end{itemize}
    
    \item \textbf{Back-Substitution to Find Remaining Variables:}
    \begin{itemize}
        \item For each solution $x_1^{(k)}$, sequentially solve for $x_j$, $j = 2, \dots, n$, using the corresponding $g_j$:
        \begin{equation}
        x_j = g_j(x_1^{(k)}), \quad j = 2, \dots, n,
        \end{equation}
        \item Ensure that each computed $x_j$ is real and positive. If not, discard the corresponding solution.
    \end{itemize}
    
    \item \textbf{Compile the Solutions:}
    \begin{itemize}
        \item Collect all tuples $(x_1^{(k)}, x_2^{(k)}, \dots, x_n^{(k)})$ that satisfy the system and are real and positive.
        \item These tuples represent all the real positive steady-state solutions of the CRN.
    \end{itemize}
\end{enumerate}

\subsection{Limitations and Considerations}\label{ssec:limitations}

While the Gr\"obner basis method offers a systematic approach to finding all real positive solutions of polynomial systems derived from CRNs, it faces significant challenges when applied to large systems. Computing the Gr\"obner basis can become computationally intensive due to the exponential increase in complexity with the number of variables and the degrees of the polynomials involved. This limitation necessitates the use of optimization techniques or alternative methods to make the computation feasible for large-scale networks.

Understanding the conjecture proposed in Subsection \ref{ssec:conjecture}--that the reduced Gr\"obner basis of such systems has a specific structured form--may lead to the development of faster algorithms. By exploiting the inherent structure of the polynomials arising from CRNs, specialized computational strategies could be devised to reduce complexity. While we cannot assert that such methods will universally apply, investigating the structural properties of these systems holds promise for enhancing computational efficiency.

Additionally, structural simplifications in the computation of Gr\"obner bases have been explored in the literature. For instance, research on reaction networks with intermediate species has demonstrated that the Gr\"obner basis of the steady-state ideal of the core network (excluding intermediates) can be extended to the full network (including intermediates) using linear algebra and appropriate monomial orderings \cite{Sadeghimanesh201974}. This approach significantly reduces computation time by decreasing the number of variables and polynomials, leveraging the network's structure to simplify calculations. Such findings underscore the potential benefits of incorporating structural insights into algorithm design to handle larger systems more effectively.

Next, we are going to illustrate as the conjecture is satisfied for two specific CNRs. The reader should be aware that the Gr\"obner basis does not always specialize (in the parameters) well, but it does generically, in the sense that the set of bad parameters is a proper Zariski subset: in particular, it has measure 0. This fact is clear while performing the Buchberger's algorithm.

\section{Application to a Gene Regulatory Network Example}\label{sec:gene_regulatory}

In this section, we aim to illustrate the conjecture presented earlier in the case of a specific gene regulatory network (GRN) \cite{Conradi2017,Siegal-Gaskins2015}. This example serves as a straightforward case where we can explicitly write the structure of the Gr\"obner basis and observe how it aligns with the conjectured form. We introduce the GRN, derive the corresponding system of polynomial equations, compute the Gr\"obner basis, and verify that its structure matches our expectations. This exercise provides initial evidence supporting the conjecture and illustrates the practical application of our method to study how the number of real positive solutions varies with respect to certain parameters.

\subsection{Description of the Gene Regulatory Network}

Consider the following gene regulatory network involving species \( X_1 \), \( X_2 \), \( P_1 \), \( P_2 \), \( X_2 P_1 \), \( P_2 P_2 \), and \( X_1 P_2 P_2 \):

\begin{align*}
X_1 \xrightarrow{\kappa_1} X_1+P_1 \qquad &P_1 \xrightarrow{\kappa_3} 0 \\
X_2 \xrightarrow{\kappa_2} X_2+P_2 \qquad &P_2 \xrightarrow{\kappa_4} 0 \\
X_2+P_1 \xrightleftharpoons[\kappa_6]{\kappa_5} X_2 P_1 \qquad &2 P_2 \xrightleftharpoons[\kappa_8]{\kappa_7} P_2 P_2\\
X_1+P_2 P_2 \xrightleftharpoons[\kappa_{10}]{\kappa_9} &X_1 P_2 P_2
\end{align*}

This network models the interactions between genes \( X_1 \) and \( X_2 \) and their products \( P_1 \) and \( P_2 \), including dimerization and complex formation processes. The reactions involve production, degradation, binding, and unbinding events governed by rate constants \( \kappa_i \).

\subsection{Derivation of the Steady-State Equations}

By applying the law of mass action and setting the time derivatives to zero (steady-state conditions), we obtain the following system of polynomial equations:

\begin{align*}
     &\quad \kappa_{10} x_7 - \kappa_9 x_1 x_6 = 0, \\
     &\quad \kappa_6 x_5 - \kappa_5 x_2 x_3 = 0, \\
     &\quad \kappa_1 x_1 - \kappa_3 x_3 + \kappa_6 x_5 - \kappa_5 x_2 x_3 = 0, \\
     &\quad -2\kappa_7 x_4^2 - \kappa_4 x_4 + \kappa_2 x_2 + 2\kappa_8 x_6 = 0, \\
     &\quad \kappa_5 x_2 x_3 - \kappa_6 x_5 = 0, \\
     &\quad \kappa_7 x_4^2 - \kappa_8 x_6 + \kappa_{10} x_7 - \kappa_9 x_1 x_6 = 0, \\
     &\quad \kappa_9 x_1 x_6 - \kappa_{10} x_7 = 0.
\end{align*}

Here, \( x_i \) represents the concentration of species corresponding to each variable:

\begin{itemize}
    \item \( x_1 \): concentration of \( X_1 \),
    \item \( x_2 \): concentration of \( X_2 \),
    \item \( x_3 \): concentration of \( P_1 \),
    \item \( x_4 \): concentration of \( P_2 \),
    \item \( x_5 \): concentration of \( X_2 P_1 \),
    \item \( x_6 \): concentration of \( P_2 P_2 \),
    \item \( x_7 \): concentration of \( X_1 P_2 P_2 \).
\end{itemize}

Additionally, the system includes conservation laws due to the total amounts of \( X_1 \) and \( X_2 \) (denoted by \( c_1 \) and \( c_2 \)):

\begin{align*}
    &\quad x_2 + x_5 = c_1, \\
    &\quad x_1 + x_7 = c_2.
\end{align*}

\subsection{Computation of the Gr\"obner Basis}

To analyze the system and find all real positive solutions, we compute the reduced Gr\"obner basis of the polynomial system using lexicographic ordering with variable prioritization:

\[
x_7 > x_6 > x_5 > x_4 > x_3 > x_1 > x_2.
\]

The computed reduced Gr\"obner basis is:

\begin{align*}
g_1(x_2): \quad & x_2^3 - c_1 x_2^2 + \left( \frac{\kappa_1 \kappa_4^2 \kappa_5 \kappa_8 \kappa_{10} c_2 + \kappa_3 \kappa_4^2 \kappa_6 \kappa_8 \kappa_{10}}{\kappa_2^2 \kappa_3 \kappa_6 \kappa_7 \kappa_9} \right) x_2 - \frac{\kappa_4^2 \kappa_8 \kappa_{10} c_1}{\kappa_2^2 \kappa_7 \kappa_9} = 0, \\
x_1-g_2(x_2): \quad & x_1 - \left( \frac{\kappa_2^2 \kappa_3 \kappa_6 \kappa_7 \kappa_9}{\kappa_1 \kappa_4^2 \kappa_5 \kappa_8 \kappa_{10}} \right) x_2^2 + \left( \frac{\kappa_2^2 \kappa_3 \kappa_6 \kappa_7 \kappa_9 c_1}{\kappa_1 \kappa_4^2 \kappa_5 \kappa_8 \kappa_{10}} \right) x_2 - c_2 = 0, \\
x_3-g_3(x_2): \quad & x_3 - \left( \frac{\kappa_2^2 \kappa_6 \kappa_7 \kappa_9}{\kappa_4^2 \kappa_5 \kappa_8 \kappa_{10}} \right) x_2^2 + \left( \frac{\kappa_2^2 \kappa_6 \kappa_7 \kappa_9 c_1}{\kappa_4^2 \kappa_5 \kappa_8 \kappa_{10}} \right) x_2 - \frac{\kappa_1 c_2}{\kappa_3} = 0, \\
x_4-g_4(x_2): \quad & x_4 - \left( \frac{\kappa_2}{\kappa_4} \right) x_2 = 0, \\
x_5-g_5(x_2): \quad & x_5 + x_2 - c_1 = 0, \\
x_6-g_6(x_2): \quad & x_6 - \left( \frac{\kappa_2^2 \kappa_7}{\kappa_4^2 \kappa_8} \right) x_2^2 = 0, \\
x_7-g_7(x_2): \quad & x_7 + \left( \frac{\kappa_2^2 \kappa_3 \kappa_6 \kappa_7 \kappa_9}{\kappa_1 \kappa_4^2 \kappa_5 \kappa_8 \kappa_{10}} \right) x_2^2 - \left( \frac{\kappa_2^2 \kappa_3 \kappa_6 \kappa_7 \kappa_9 c_1}{\kappa_1 \kappa_4^2 \kappa_5 \kappa_8 \kappa_{10}} \right) x_2 = 0.
\end{align*}

This structure confirms that the system has been effectively triangularized, allowing for sequential solving starting from \( x_2 \). It also aligns with the central idea in Conjecture \ref{conj}, although the variable priorities here differ from the one proposed in the conjecture.

We tested other variable priorities and consistently found a triangular structure, with no notable impact on computation time or complexity for this example, so those results are not reported here; however, we do explore different variable priorities in the more complex model discussed in Section \ref{sec:wnt}, where we notice significant differences.


\subsection{Analysis of the Number of Real Positive Solutions}

To determine the number of real positive solutions, we focus on solving \( g_1(x_2) = 0 \) and then use back-substitution to find the remaining variables. We fix the rate constants \( \kappa_i \) and consider a grid of positive values \( (c_1, c_2) \), which may represent total concentrations of certain proteins or genes. By varying \( c_1 \) and \( c_2 \), we can simulate how changes in gene expression levels or external stimuli affect the network's behavior, shifting from a unique steady state to multistability. Note that this analysis is missing a parameter sensitivity study: we assume fixed rate constants \( \kappa_i \), while in reality these constants may vary due to environmental factors or mutations.

By analyzing the plots, we can identify regions in the \( (c_1, c_2) \) plane where the system has:
\begin{itemize}
    \item One positive real solution: indicating a unique steady state.
    \item Three positive real solutions: indicating the possibility of multiple steady states, which could correspond to multistability in the biological system.
\end{itemize}
This information is valuable for understanding the conditions under which the gene regulatory network exhibits different dynamic behaviors.

Figure \ref{fig:grn_k} illustrates the regions with different numbers of positive real solutions for two random sets of rate constants \( \kappa \):
\[ \begin{array}{ll}
    & \kappa_1   = 0.551, \kappa_2   = 0.708, \kappa_3   = 0.291, \kappa_4   = 0.511, \kappa_5   = 0.893, \\
    &\kappa_6   = 0.896, \kappa_7   = 0.126, \kappa_8   = 0.207, \kappa_9   = 0.052, \kappa_{10}= 0.441 \qquad 
    \end{array}\text{(left)}
\]
\[\begin{array}{ll}
    &\kappa_1   = 0.815, \kappa_2   = 0.906, \kappa_3   = 0.127, \kappa_4   = 0.913, \kappa_5   = 0.632, \\
    &\kappa_6   = 0.098, \kappa_7   = 0.279, \kappa_8   = 0.547, \kappa_9   = 0.958, \kappa_{10}= 0.965 \qquad 
    \end{array}\text{(right)}.
\]

\begin{figure}
\centering
\includegraphics[width = 0.48\textwidth]{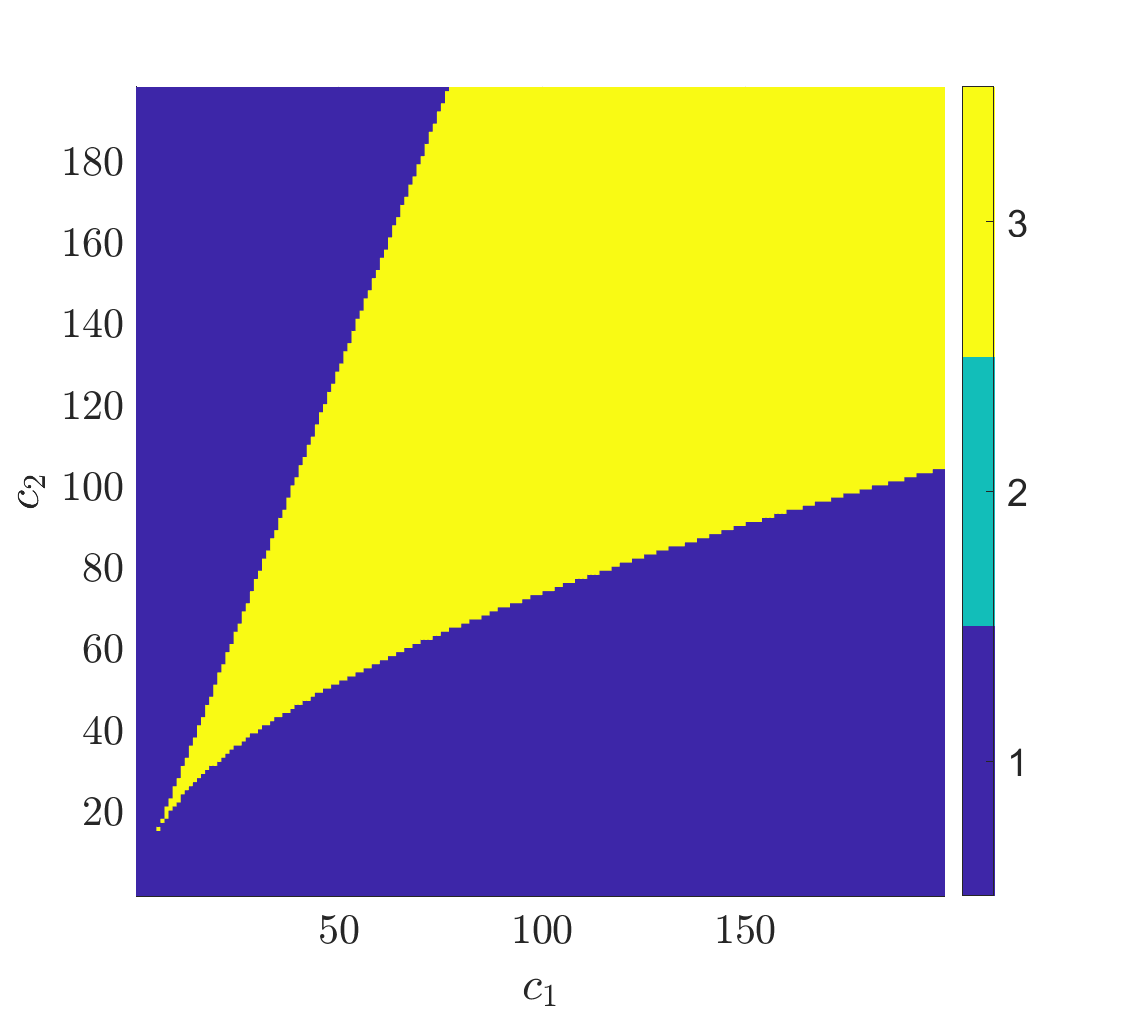}
\includegraphics[width = 0.47\textwidth]{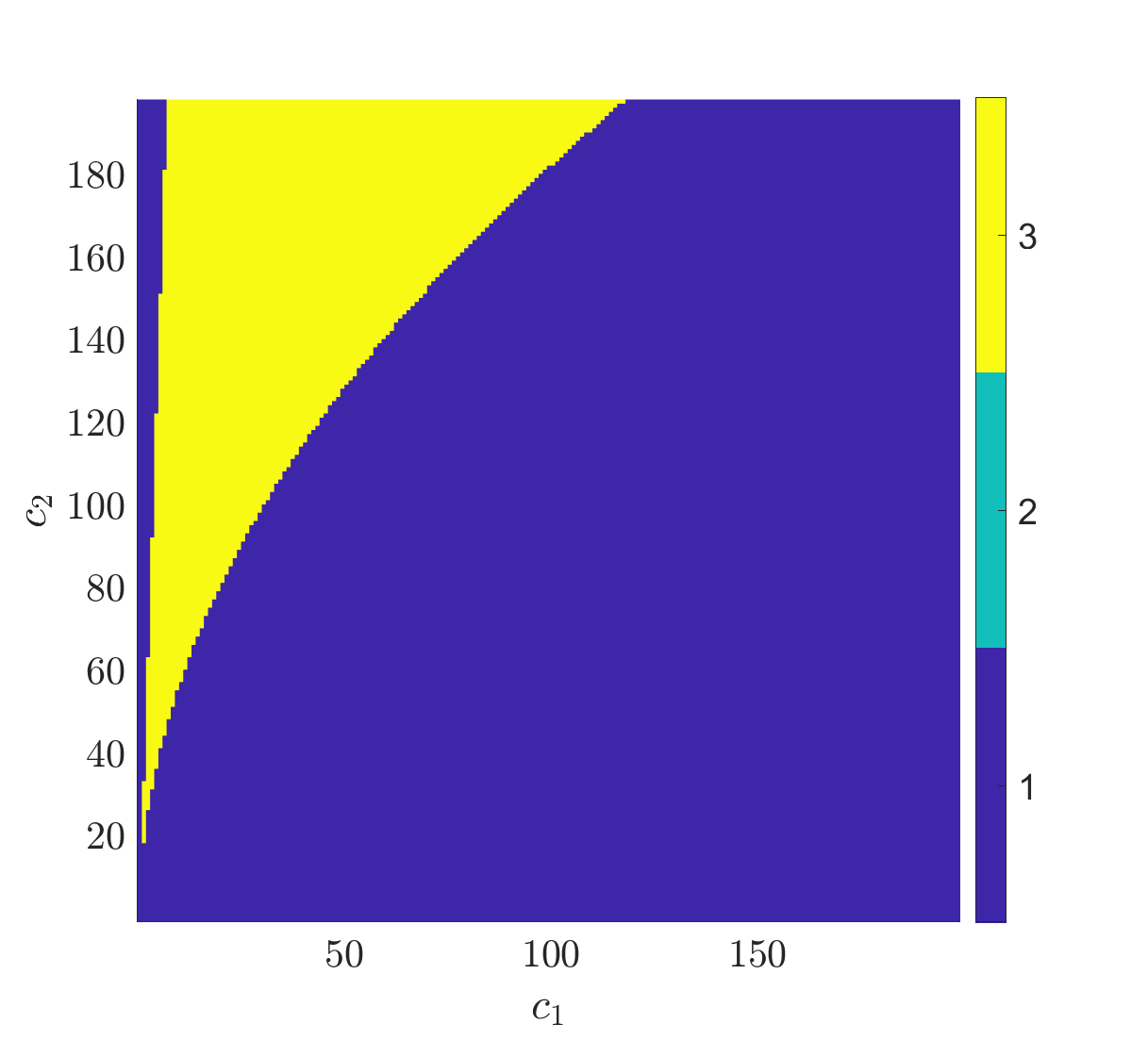}
\caption{Number of positive real solutions for two random \( \kappa \) vectors.}\label{fig:grn_k}
\end{figure}

\section{Application to the Wnt signaling pathway}\label{sec:wnt}

The Wnt signaling pathway is a highly conserved regulatory network that plays a crucial role in embryonic development, tissue homeostasis, and cell fate determination \cite{Gross201621,Maclean2016405}. Dysregulation of this pathway is implicated in various diseases, including cancer. Understanding its steady-state behavior is therefore of fundamental interest. In this section, we demonstrate how the Gr\"obner basis methodology introduced earlier can be employed to systematically determine all real positive steady states of a polynomial system derived from a Wnt pathway model under mass-action kinetics.

The simplified Wnt model is composed by 19 distinct species whose concentrations and interactions define the system dynamics.
Each variable $x_i$ in the model corresponds to the concentration of one such species. 
These include various states of central pathway components - dishevelled, the destruction complex, phosphatase, and $\beta$-catenin - distributed across different cellular compartments, as well as transcription factors and intermediate complexes formed during signaling:
\begin{itemize}
\item Dishevelled states ($x_1$, $x_2$, $x_3$) are present in the cytoplasm in both inactive ($x_1$) and active ($x_2$) forms. It can also exist in an active nuclear form ($x_3$). Transition among these states and compartments is integral to pathway regulation, particularly after Wnt stimulation.
\item The destruction Complex ($x_4$, $x_5$, $x_6$, $x_7$) is composed of APC, Axin, and GSK3$\beta$, and is crucial for controlling $\beta$-catenin levels. It cycles between active and inactive forms and can reside in the cytoplasm ($x_4$ active, $x_5$ inactive) or the nucleus ($x_6$ active, $x_7$ inactive). 
\item Phosphatase ($x_8$, $x_9$) modifies the destruction complex and other proteins. It is present in both the cytoplasm ($x_8$) and the nucleus ($x_9$).
\item $\beta$-catenin ($x_{10}$, $x_{11}$) is a key effector of Wnt signaling, and is found in the cytoplasm ($x_{10}$) and can shuttle to the nucleus ($x_{11}$).
\item Transcription Factor ($x_{12}$) and Transcription Complex ($x_{13}$): Within the nucleus, TCF ($x_{12}$) can bind to $\beta$-catenin to form a transcriptionally active complex ($x_{13}$), which directly influences Wnt target gene expression.
\item Intermediate Complexes ($x_{14}$ through $x_{19}$): Several transient complexes arise during signal transduction. 
\end{itemize}

By monitoring these 19 species and their 31 associated reactions—each governed by a distinct mass-action rate constant $k_i$—the Wnt shuttle model represents the interplay between extracellular Wnt signals, cytoplasmic and nuclear responses, and gene transcription. The resulting polynomial ODE system captures the pathway’s dynamic behavior \cite{Gross201621}.

By writing down the mass-action ODEs for each species and setting the time derivatives to zero, we obtain a system of polynomial equations. Additionally, the conservation laws are included as polynomial constraints:
\begin{equation}\label{eq:wnt_cons_laws}
\begin{split}
& c_{1}=x_{1}+x_{2}+x_{3}+x_{14}+x_{15} \\
& c_{2}=x_{4}+x_{5}+x_{6}+x_{7}+x_{14}+x_{15}+x_{16}+x_{17}+x_{18}+x_{19} \\
& c_{3}=x_{8}+x_{16} \\
& c_{4}=x_{9}+x_{17} \\
& c_{5}=x_{12}+x_{13}
\end{split}
\end{equation}

According to the algorithm in Subsection \ref{ssec:algorithm}, the key to finding all real positive solutions is to compute the reduced Gr\"obner basis of the ideal generated by the steady-state equations. As before, we use a lexicographic monomial ordering. The lex order helps in obtaining a triangularized system after the Gr\"obner basis computation. 

However, using different variable priorities within the lexicographic order may impact the time required for computing the Gr\"obner basis. To demonstrate this we run the following experiment. We considered the WNT signaling pathways and we defined the network parameters (namely the rate constants $\kappa$ and the conservation laws constants $c_i$) as in Theorem 4.1 of \cite{Gross201621} in order to obtain multistability.  Only the constant $c_4$ associated to the fourth conservation law in Eq. (\ref{eq:wnt_cons_laws}) was considered as parameter. Fig. \ref{fig:CPU_times} shows the computational time required for computing the Gr\"obner basis for 50 different variable priorities listed in the first column of Table \ref{tab:CPU_times}, Appendix \ref{sec:App}. In 54\% of the considered cases the Gr\"obner basis was computed in less than 100 seconds, while for 16\% of the cases it took more than 2 hours with a maximum of about 9.5 hours. A further analysis, not shown here, demonstrates that the variable priorities associated with a higher computational time return Gr\"obner bases whose polynomial coefficients exhibit a more intricate dependence on the parameter $c_4$.

By scanning across different parameter values, one can identify conditions under which the system transitions from having a single stable positive steady state to multiple steady states. Such bistability or multistability could have biological significance, for instance, in representing switch-like behavior that controls cell fate decisions. In particular, in Figure \ref{fig:wnt_n_sol}, we can identify regions in the \( (c_1, c_2) \) plane and in the \( (c_3, c_4) \) plane where the system has one positive real solution, indicating a unique steady state, or three positive real solutions, which could correspond to multistability in the biological system.

Finally, we highlight that for the Wnt pathway, all computed Gr\"obner bases indeed conform to the structural pattern anticipated by Conjecture \ref{conj}.

\begin{figure}
\centering
\includegraphics[width = 0.47\textwidth]{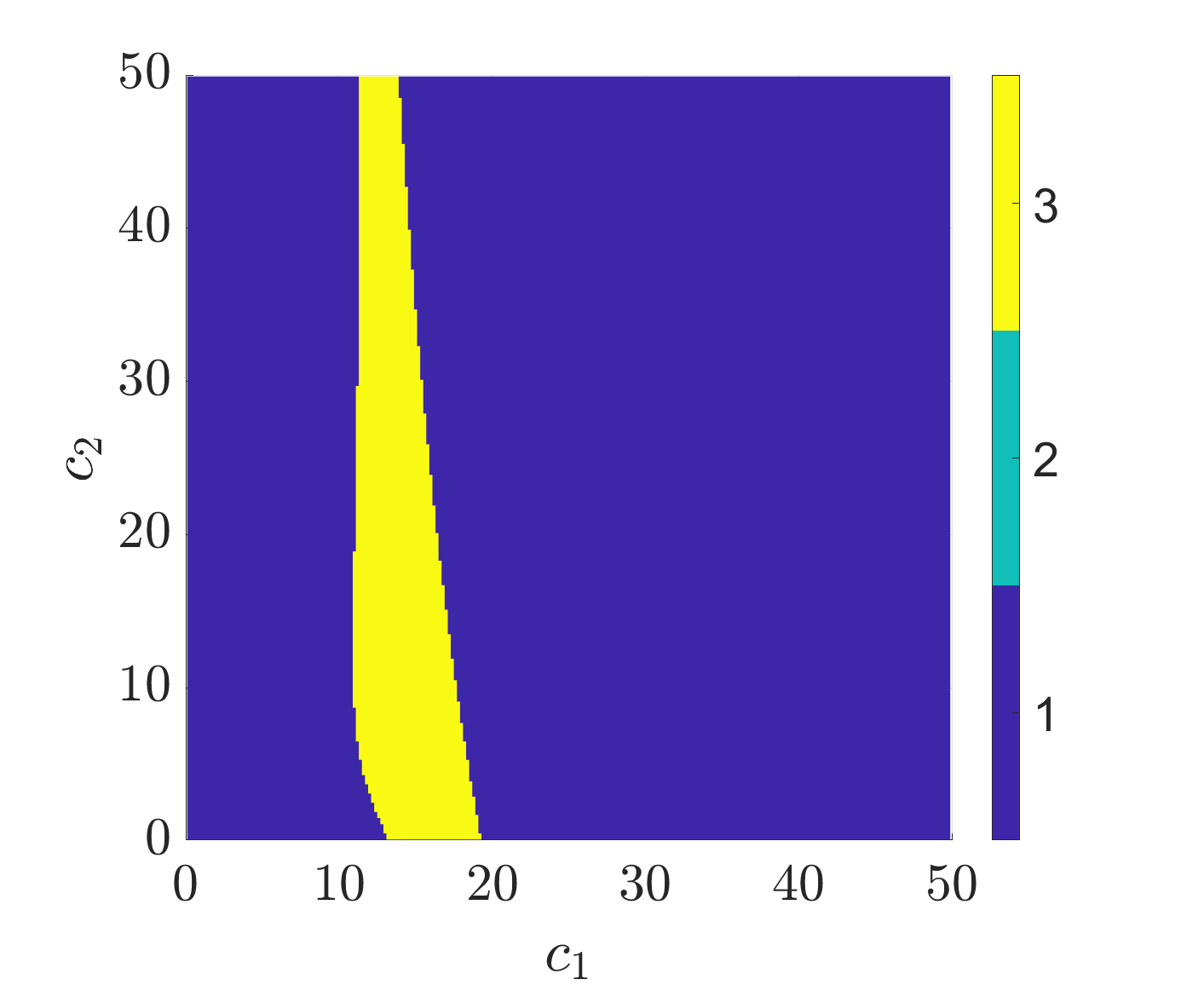}
\includegraphics[width = 0.49\textwidth]{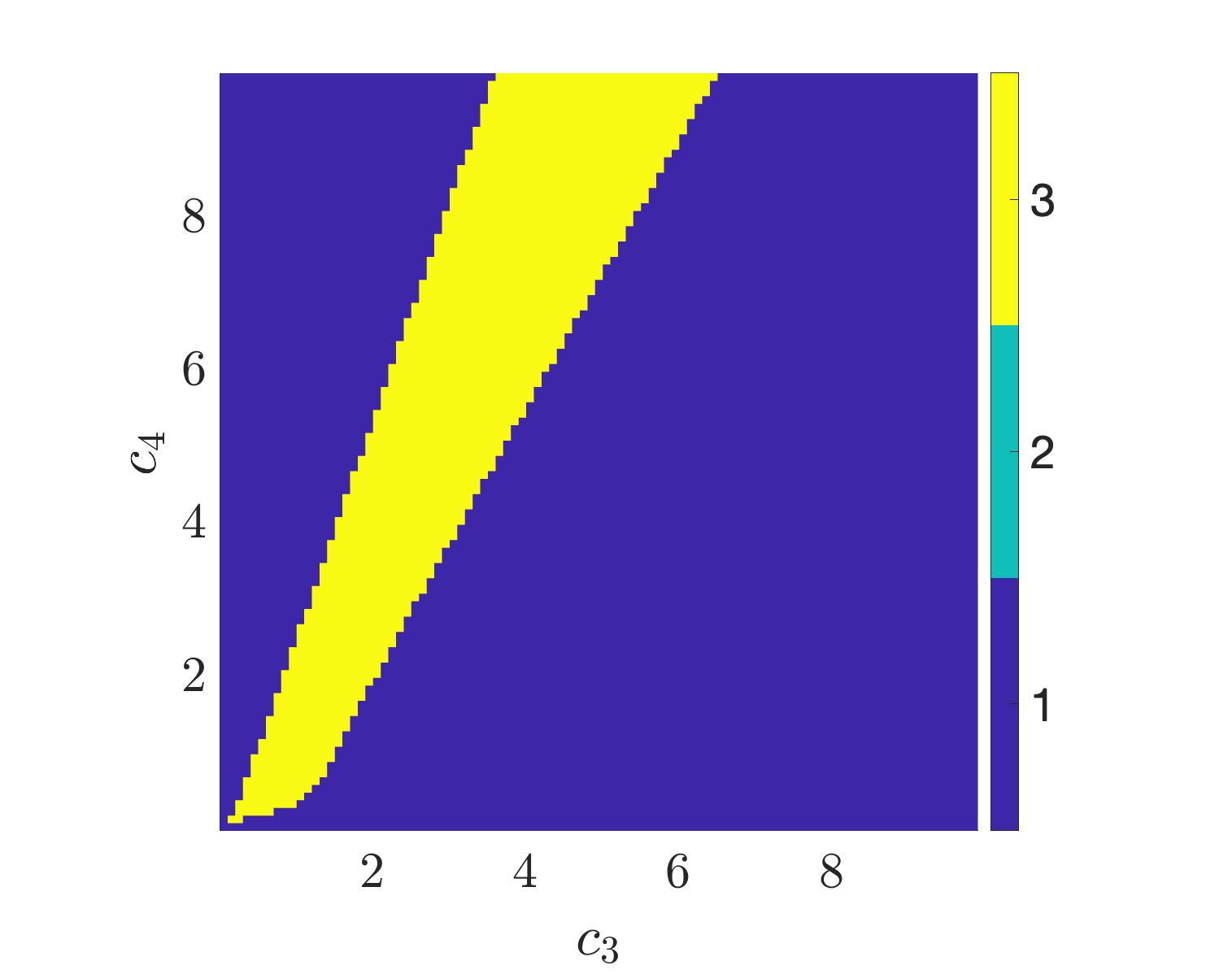}
\caption{Wnt pathway: number of positive real solutions varying $c_1$ and $c_2$ (left) and $c_3$ and $c_4$ (right).}\label{fig:wnt_n_sol}
\end{figure}

\begin{figure}
\centering
\includegraphics[width = 0.8\textwidth]{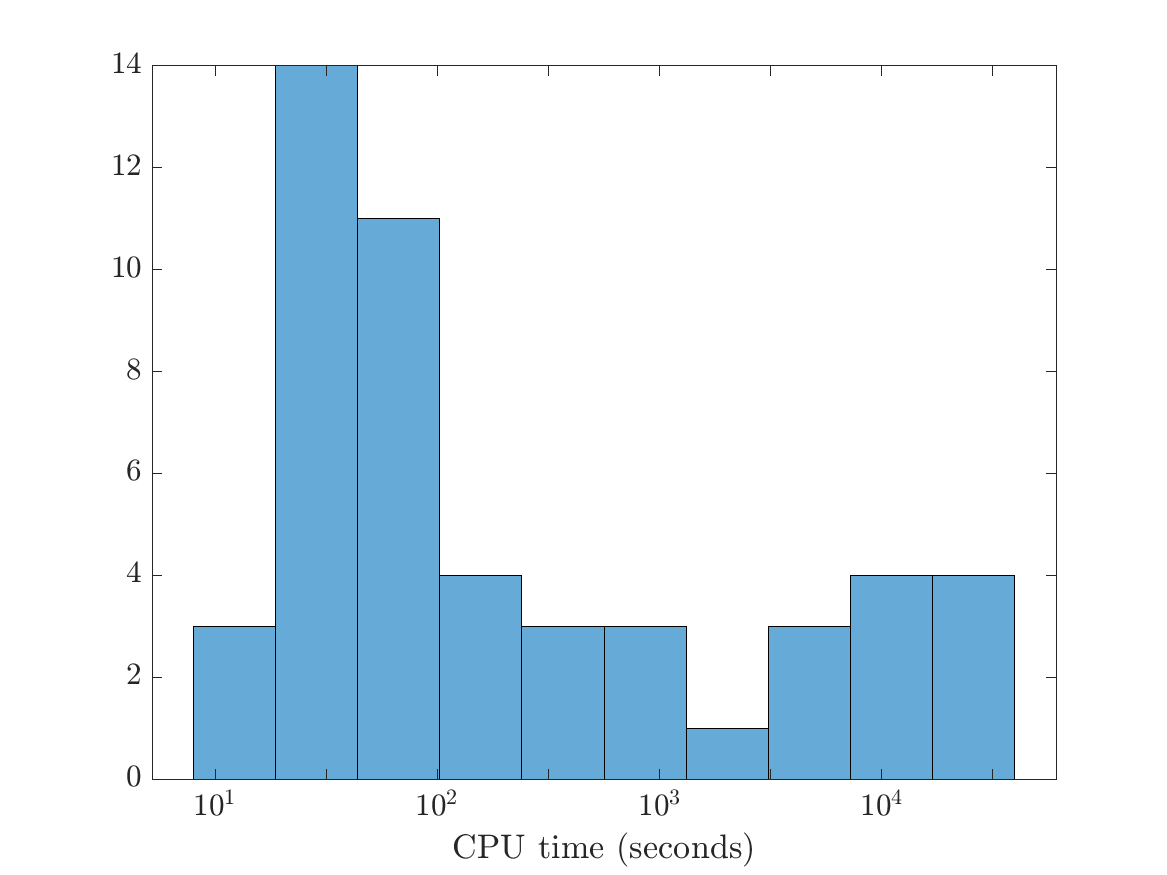}
\caption{Histogram for the CPU times in logarithmic scale for the computation of the Gr\"obner basis with symbolic $c_4$ with lexicographic ordering varying the variable priorities.}\label{fig:CPU_times}
\end{figure}

\section{Conclusions and future work}\label{sec:conclusions}
This work demonstrates how Gr\"obner bases can systematically reveal all real positive steady states in polynomial systems arising from chemical reaction networks governed by mass-action kinetics. By focusing on a lexicographic ordering, we achieve a triangular form that permits sequential solution. The conjectured structure of the reduced Gr\"obner basis aligns well with the tested networks in \cite{M2Repository} such as gene regulation and Wnt signaling models.

In future research, we will refine the conditions under which the conjectured triangular Gr\"obner basis structure holds and streamline the present Jacobian‑criterion condition, thereby establishing a rigorous but readily deployable framework for this structured Gr\"obner approach. We also aim to accelerate Gr\"obner basis computations by leveraging the knowledge on the Gr\"obner basis structure and by defining a variable order that respects each network’s stoichiometric organization--for instance, placing a ``key'' species last so that the associated variable is the one that appears in the univariate polynomial.

In addition, we plan to explore large-scale biological networks and investigate the numerical stability of parametric Gr\"obner bases, seeking parameter sets that guarantee unique solutions. Such uniqueness is particularly valuable for developing a robust drug-response simulator, where the effects of a therapeutic intervention can be reliably predicted based on well-defined system parameters, continuing the work in \cite{Sommariva2023}.

\section*{Declarations}

\subsection*{Acknowledgments and Funding}
This work was carried out within the framework of the project HubLife Science –
Digital Health (LSH-DH) PNC-E3-2022-23683267 - Progetto DHEAL-COM – CUP:D33C22001980001,
founded by Ministero della Salute within ``Piano Nazionale Complementare al PNRR Ecosistema Innovativo della Salute - Codice univoco investimento: PNC-E.3''. The research was supported in part by the MUR Excellence Department Project awarded to Dipartimento di Matematica, Università di Genova, CUP D33C23001110001. FB, PF, MP and SS are members of “Gruppo Nazionale per il Calcolo Scientifico" (INdAM-GNCS). MV is supported by the Italian PRIN2020, Grant number 2020355B8Y, ``Squarefree Gr\"obner
degenerations, special varieties and related topics''. 

\subsection*{Conflict of interest}
On behalf of all authors, the corresponding author states that there is no conflict of interest.

\subsection*{Author contribution}
FB PF SS MP designed the study. PF run the simulations. MV PF and SS interpreted the results and wrote the original draft. All the authors revised and wrote the final version of the manuscript.

\begin{appendices}
\section{CPU times analysis}\label{sec:App}
In this appendix, we provide the raw data collected for the computational times discussed in Section \ref{sec:wnt}. All computations were performed using Macaulay2 on the University of Melbourne’s cloud-based server.

\begin{table}[]
\footnotesize
    \centering
    \begin{tabular}{l|r}
         Variable Priorities & CPU Times (s) \\
         \hline
x2,x9,x18,x13,x7,x4,x5,x12,x10,x1,x14,x15,x3,x19,x8,x16,x17,x6,x11          &          2934.33           \\
x18,x3,x15,x4,x9,x11,x5,x17,x7,x2,x13,x10,x6,x8,x12,x1,x16,x14,x19          &          5082.17           \\
x6,x16,x19,x15,x7,x5,x4,x8,x18,x12,x2,x10,x3,x11,x1,x17,x14,x13,x9          &          1293.69           \\
x5,x17,x3,x15,x1,x8,x11,x7,x2,x14,x16,x9,x13,x4,x10,x18,x19,x6,x12      	&			22154.60         \\
x12,x13,x3,x2,x7,x10,x9,x16,x6,x15,x1,x14,x18,x17,x11,x4,x19,x8,x5       	&			  74.99             \\
x1,x7,x4,x19,x6,x11,x14,x13,x8,x16,x2,x18,x12,x10,x17,x5,x3,x9,x15          &            16.59             \\
x10,x6,x19,x3,x11,x13,x8,x14,x15,x2,x5,x7,x17,x4,x1,x16,x18,x9,x12          &            17113.00       \\
x17,x5,x18,x19,x10,x11,x1,x7,x4,x13,x12,x3,x8,x2,x16,x9,x6,x15,x14       	&			  12.37             \\
x5,x14,x11,x7,x13,x1,x16,x4,x9,x15,x19,x8,x18,x12,x10,x6,x2,x3,x17          &            24.48             \\
x18,x6,x7,x4,x14,x9,x19,x12,x8,x10,x3,x5,x15,x13,x11,x1,x17,x2,x16       	&			  21.23             \\
x5,x13,x8,x17,x6,x9,x15,x2,x16,x14,x10,x1,x3,x12,x19,x7,x18,x4,x11          &          8616.78           \\
x7,x8,x6,x19,x18,x4,x17,x10,x9,x14,x5,x12,x11,x16,x15,x13,x1,x3,x2          &           387.05            \\
x6,x15,x13,x1,x5,x10,x11,x8,x18,x4,x2,x9,x17,x12,x14,x3,x19,x7,x16          &            53.95             \\
x8,x1,x16,x10,x14,x17,x11,x18,x9,x6,x4,x19,x15,x5,x7,x2,x13,x3,x12          &            20084.30       \\
x19,x8,x14,x15,x11,x10,x5,x9,x13,x3,x18,x12,x4,x2,x1,x16,x17,x6,x7          &            54.24             \\
x9,x11,x2,x17,x18,x19,x4,x16,x12,x1,x13,x15,x7,x5,x6,x8,x3,x10,x14          &            16.67             \\
x6,x7,x13,x12,x16,x2,x8,x3,x18,x1,x5,x19,x10,x4,x15,x14,x17,x11,x9          &            79.59             \\
x12,x4,x6,x16,x13,x17,x10,x8,x15,x19,x1,x2,x7,x9,x14,x3,x5,x11,x18          &            14172.50       \\
x3,x4,x18,x6,x1,x17,x5,x19,x15,x14,x12,x10,x11,x16,x2,x7,x8,x9,x13          &            34360.30       \\
x12,x11,x18,x19,x1,x6,x8,x4,x7,x13,x5,x2,x16,x10,x14,x3,x9,x15,x17          &            23.88             \\
x3,x9,x12,x16,x2,x15,x10,x8,x18,x13,x4,x6,x19,x7,x11,x14,x1,x5,x17          &            19.13             \\
x19,x7,x8,x5,x2,x3,x1,x11,x18,x6,x17,x13,x16,x10,x14,x12,x4,x15,x9          &           388.04            \\
x16,x4,x9,x13,x15,x7,x3,x2,x19,x5,x12,x17,x18,x8,x11,x10,x1,x14,x6          &           138.06            \\
x7,x17,x15,x9,x6,x12,x19,x18,x5,x1,x10,x14,x11,x3,x13,x2,x8,x16,x4       	&			 102.75            \\
x2,x4,x15,x5,x17,x18,x7,x3,x10,x13,x1,x8,x14,x11,x12,x6,x19,x16,x9          &			  78.08             \\
x6,x8,x9,x17,x3,x16,x2,x19,x13,x10,x5,x7,x4,x12,x11,x14,x18,x15,x1          &            94.63             \\
x7,x3,x5,x6,x8,x9,x18,x10,x19,x13,x14,x12,x11,x2,x17,x4,x16,x1,x15          &			  28.34             \\
x9,x17,x10,x11,x15,x8,x1,x12,x6,x13,x2,x3,x14,x18,x5,x7,x19,x16,x4          &           239.29            \\
x5,x15,x3,x14,x19,x6,x16,x8,x13,x10,x7,x17,x9,x4,x1,x12,x18,x11,x2          &           100.70            \\
x3,x10,x14,x16,x1,x18,x9,x7,x17,x5,x2,x19,x4,x13,x15,x12,x6,x8,x11          &          4513.55           \\
x12,x1,x18,x6,x9,x10,x5,x19,x11,x16,x14,x4,x17,x7,x13,x3,x2,x15,x8          &			  20.40             \\
x11,x4,x9,x17,x3,x8,x10,x13,x15,x5,x2,x1,x12,x6,x18,x19,x14,x16,x7          &            68.41             \\
x12,x17,x15,x18,x11,x3,x4,x14,x1,x9,x6,x19,x2,x5,x10,x8,x13,x7,x16          &           307.31            \\
x19,x1,x6,x11,x10,x4,x14,x13,x17,x18,x12,x7,x15,x9,x5,x16,x2,x8,x3			&			  34.28             \\
x13,x11,x4,x10,x14,x6,x8,x7,x2,x18,x15,x3,x12,x16,x1,x19,x5,x17,x9          &            42.33             \\
x4,x12,x5,x14,x8,x6,x18,x13,x3,x7,x10,x9,x1,x19,x11,x2,x17,x16,x15          &            24.22             \\
x10,x13,x3,x7,x8,x11,x17,x5,x2,x16,x4,x9,x14,x12,x18,x6,x15,x1,x19          &          7014.19           \\
x15,x10,x11,x19,x4,x3,x5,x13,x18,x2,x6,x14,x12,x16,x9,x17,x1,x8,x7          &            25.34             \\
x8,x14,x18,x6,x9,x17,x5,x3,x4,x11,x12,x19,x1,x2,x7,x15,x16,x13,x10          &          9120.25           \\
x5,x15,x2,x1,x9,x3,x13,x4,x17,x14,x7,x11,x18,x6,x19,x12,x16,x10,x8          &           229.22            \\
x18,x8,x5,x15,x1,x3,x14,x13,x2,x19,x11,x7,x4,x10,x17,x12,x9,x6,x16          &            29.73             \\
x6,x9,x8,x13,x4,x2,x7,x5,x3,x19,x16,x10,x11,x1,x15,x14,x12,x18,x17          &           834.58            \\
x10,x4,x7,x11,x12,x15,x18,x6,x13,x19,x2,x16,x1,x9,x5,x14,x8,x17,x3          &            37.23             \\
x15,x13,x3,x11,x4,x7,x9,x8,x18,x10,x19,x5,x12,x6,x16,x17,x2,x14,x1          &            30.28             \\
x16,x1,x14,x6,x8,x12,x19,x10,x3,x4,x9,x11,x18,x2,x13,x5,x7,x15,x17          &            90.51             \\
x1,x10,x13,x5,x19,x6,x17,x16,x4,x12,x18,x11,x7,x8,x3,x14,x15,x2,x9          &            20.40             \\
x16,x15,x18,x1,x7,x10,x4,x3,x13,x2,x14,x9,x5,x12,x19,x11,x17,x8,x6          &            56.25             \\
x16,x12,x15,x10,x17,x5,x18,x7,x8,x3,x1,x19,x2,x13,x9,x4,x14,x11,x6          &          8789.40           \\
x1,x17,x13,x19,x16,x8,x2,x3,x7,x15,x4,x14,x11,x12,x10,x9,x18,x6,x5			&			  63.60             \\
x8,x2,x10,x14,x1,x16,x15,x3,x6,x11,x9,x12,x18,x5,x17,x13,x19,x4,x7          &           963.57            \\
    \end{tabular}
    \caption{CPU times for Gr\"obner basis computations with symbolic \(c_4\) under different lexicographic monomial orderings, where each ordering ranks the variables in descending priority.}
    \label{tab:CPU_times}
\end{table}
\end{appendices}

\bibliography{references}

\end{document}